\documentclass{article}

\usepackage{amssymb,amsmath,amsthm}
\usepackage[all]{xy} \SelectTips{eu}{12}

\newtheorem{theorem}{Theorem}[section]
\newtheorem{proposition}[theorem]{Proposition}
\newtheorem{lemma}[theorem]{Lemma}

\theoremstyle{definition}
\newtheorem{definition}[theorem]{Definition}
\newtheorem{remark}[theorem]{Remark}
\newtheorem{example}[theorem]{Example}

\newcommand{\arsto}{\rightsquigarrow} 
\newcommand{\mo}{\xymatrix@C=1pc{\ar[r]&}} 
\newcommand{\monomo}{\xymatrix@C=1pc{\ar@{>->}[r]&}} 
\newcommand{\parmo}{\xymatrix@C=1pc{\ar@{-^{>}}[r]&}} 
\newcommand{\monoparmo}{\xymatrix@C=1pc{\ar@{>-^{>}}[r]&}} 
\newcommand{\hetmo}{\xymatrix@C=1pc{\ar@{^{<}-^{>}}[r]&}} 

\newcommand{\spa}[5]{#2:#1\leftarrow#3\rightarrow#5:#4} 
\newcommand{\catC}{\mathcal{C}}
\newcommand{\catgraph}{\mathbf{Graph}} 
\newcommand{\catgraphp}{\mathbf{Graph^p}} 
\newcommand{\catgraphm}{\mathbf{Graph_m}} 
\newcommand{\catgraphi}{\mathbf{Graph_\subseteq}} 
\newcommand{\catgraphpm}{\mathbf{Graph^p_m}} 
\newcommand{\cattermgraph}{\mathbf{TermGraph}} 
\newcommand{\cattermgraphm}{\mathbf{TermGraph_m}} 
\newcommand{\catD}{\mathcal{D}}
\newcommand{\catM}{\mathcal{M}}
\newcommand{\catL}{\mathcal{M}_L}
\newcommand{\catR}{\mathcal{M}_R}
\newcommand{\catP}{\mathcal{P}}

\newcommand{\funL}{\mathcal{L}} 
\newcommand{\funR}{\mathcal{R}} 
\newcommand{\funS}{\mathbf{S}} 
\newcommand{\funSrc}{\mathit{Src}}
\newcommand{\funTgt}{\mathit{Tgt}}
\newcommand{\po}{\mathit{PO}} 
\newcommand{\spo}{\mathit{SPO}} 
\newcommand{\hpo}{\mathit{HPO}}
\newcommand{\dpo}{\mathit{DPO}}   
\newcommand{\poc}{\mathit{POC}}  
\newcommand{\sqpo}{\mathit{SqPO}}   
\newcommand{\fpbc}{\mathit{FPBC}}  
\newcommand{\rgr}{\mathit{RGC}}  
\newcommand{\lgr}{\mathit{LGC}}  
 
\newcommand{\id}{\mathrm{id}} 
\newcommand{\gr}[2]{\Lambda_{#1}(#2)} 
\newcommand{\coslice}{\!\downarrow\!} 
\newcommand{\slice}{\!\uparrow\!} 
\newcommand{\rs}{\mathrm{RS}} 
\newcommand{\dom}{\mathrm{Dom}} 

\newcommand{\safe}[1]{}

\begin{document}

\title{Categorical Abstract Rewriting Systems and \\ 
       Functoriality of Graph Transformation} 
\author{Dominique Duval \\  
 LJK -- Universit\'e de Grenoble\\
B. P. 53,  F-38041 Grenoble, France \\ 
{\sf Dominique.Duval@imag.fr} \\   \and Rachid Echahed \\
LIG -- Universit\'e de Grenoble\\
B. P. 53,  F-38041 Grenoble, France\\
{\sf Rachid.Echahed@imag.fr} \and 
        Fr\'ed\'eric Prost \\
LIG -- Universit\'e de Grenoble\\
B. P. 53,  F-38041 Grenoble, France\\
{\sf Frederic.Prost@imag.fr} \\}

\safe{
\autlabel{1} \email{Dominique.Duval@imag.fr} \\ 
 Laboratoire LJK -- Universit\'e de Grenoble\\
B. P. 53,  F-38041 Grenoble, France 
 
\autlabel{2} \email{Rachid.Echahed@imag.fr} \hspace{0.2cm} \autlabel{3}  \email{Frederic.Prost@imag.fr} \\
Laboratoire LIG -- Universit\'e de Grenoble\\
B. P. 53,  F-38041 Grenoble, France
}


\date{}
\maketitle

\begin{abstract} Abstract rewriting systems are often defined as binary
  relations over a given set of objects. 
  In this paper, we
  introduce a new notion of abstract rewriting system in the framework of
  categories. Then, we define the \emph{functoriality} property of
  rewriting systems. This property is sometimes called \emph{vertical
    composition}. We show that most graph transformation systems
  are functorial and provide a counter-example of graph transformation
  system which is not functorial.  
\end{abstract}

\section{Introduction}

Various properties of rewriting systems can be defined on an abstract
level by using the notion of abstract rewriting systems (see e.g.,
\cite{BaaderN98}).
In this paper we focus on categorical rewriting systems, that is to
say rewriting systems defined by means of category theory, and we
define them in an abstract manner. We consider rule-based frameworks
in which the rewrite step is defined relatively to a match.  The aim
is to be able to reason abstractly about rewriting systems which are
defined categorically.  There are many such systems which underly
graph transformation, following the seminal work of \cite{EhrigPS73}.
In general, a graph rewriting system consists of a set of graph
rewrite rules with a left-hand side $L$ and a right-hand side $R$
(where both are graphs).  When a graph rewrite rule is applied to an
instance of the graph $L$ in a graph $L_1$, it replaces this instance
of $L$ by an instance of $R$, resulting in a new graph $R_1$.  We
introduce categorical rewriting systems in section~\ref{sec:crs}, they
provide an abstract framework for dealing with such notions of rewrite
rules, instances and rewrite steps.
Moreover, in a graph rewriting system, usually the given graph $L_1$
and the modified graph $R_1$ can be seen as the left-hand side and
right-hand side of a new rule, from which the process can be repeated.
Then the functoriality problem appears: from an instance of $L$ in
$L_1$ and an instance of $L_1$ in $L_2$, do we get the same graph
$R_2$ when proceeding in two steps as when proceeding in one step?
The functoriality property is sometimes called the vertical
composition.  It is similar to the property of contextual closure of
term rewriting systems.  A recent work of M. L\"owe \cite{Lowe10}
adresses a similar issue in a different setting in which matches are
spans instead of morphisms.  In section~\ref{sec:graph} we check that
the functoriality property holds for many usual algebraic graph
transformation approaches like double pushouts (DPO)
\cite{CorradiniMREHL97}, single pushouts (SPO) \cite{Lowe93},
sesqui-pushouts (SqPO) \cite{CorradiniHHK06} and heterogeneous
pushouts (HPO) \cite{DuvalEP09}.
Then in section~\ref{sec:garbage} we look at garbage removal 
as a categorical rewriting system,
in two different ways. This yields a categorical rewriting system 
which is functorial, and another one which is not functorial. 
We refer to \cite{MacLane} for categorical notions:
mainly commutative diagrams, functors, pushouts and pullbacks,
comma categories. 
The class of objects of a category $\catC$ is denoted as $|\catC|$.
A subcategory $\catM$ of a category $\catC$ 
is called a wide subcategory of $\catC$ 
if it has the same objects as $\catC$.

\section{Categorical rewriting systems} 
\label{sec:crs}

\subsection{Definition of categorical rewriting systems} 
\label{subsec:crs-defi}

\begin{definition}
\label{definition:crs-defi}
A \emph{categorical rewriting system} 
$(\spa{\catL}{\funL}{\catP}{\funR}{\catR},\funS)$
is made of a span of categories 
  $$ \xymatrix@=1pc{
  & \ar[ld]_{\funL} \catP \ar[rd]^{\funR} & \\
  \catL & & \catR \\ 
  } $$
and a family of partial functions 
  $$ \funS=(\funS_\rho)_{\rho\in|\catP|} $$ 
where for each object $\rho$ in $\catP$, 
the partial function $\funS_\rho$, 
from the set of morphisms in $\catL$ with source $\funL(\rho)$ 
to the set of morphisms in $\catP$ with source $\rho$, 
is such that $\funL(\funS_\rho(f))=f$ for every $f$ in the domain 
of $\funS_\rho$.
The objects of $\catP$ are the \emph{rewrite rules} or \emph{productions},
the morphisms of $\catL$ and $\catR$ are the 
left-hand side and right-hand side \emph{matches},
and the partial function $\funS_\rho$ is the \emph{rewriting process} function 
with respect to $\rho$; its domain is denoted as $\dom(\funS_\rho)$.
Given a rule $\rho$, the \emph{rewrite step applying $\rho$} 
is the partial function 
from the set of morphisms in $\catL$ with source $\funL(\rho)$ 
to the set of morphisms in $\catR$ with source $\funR(\rho)$ 
which maps every match $f$ in $\dom(\funS_\rho)$
to the match $g=\funR(\funS_\rho(f))$. 
The target $R_1$ of $g$ may be called the \emph{derived} object, 
with respect to the rule $\rho$ and the match $f$. 
\end{definition}

\begin{remark}
Many categorical rewriting systems are such that $\catL=\catR$,
then this category is denoted as $\catM$.
For the interested 
reader we refer to \cite{DuvalEP11} as an example of a rewriting system
defined by composition of rewriting systems (such composition is defined 
in section \ref{subsec:crs-comp}) in which $\catL \not =\catR$. 
\end{remark}

\begin{remark}
Each categorical rewriting system with $\catL=\catR=\catM$ 
determines an \emph{abstract rewriting system} 
on the objects of $\catM$, i.e., a binary relation $\arsto$ on $|\catM|$,
defined by $L\arsto R$ if and only if there is some $\rho$ in $\catP$
such that $L=\funL(\rho)$ and $R=\funR(\rho)$. 
\end{remark}

In a categorical rewriting system, 
the matches introduce a  ``vertical dimension'',
in addition to the ``horizontal dimension'' provided by the rules.
A rule $\rho$ with $\funL(\rho)=L$ and $\funR(\rho)=R$ 
is denoted as $\rho: L \arsto R$. 
It should be noted that, although $\rho$ is an \emph{object} in 
the category $\catP$, it is 
represented as an \emph{arrow} from its left-hand side $L$ 
to its right-hand side $R$; 
this refers to the usual notation for rewriting systems.  
Whenever $\catP$ is a category of arrows, 
it may happen that $\rho$ actually is a morphism in some category 
$\catD$, with either $\rho:L\to R$ 
(as in sections~\ref{subsec:graph-spo} and~\ref{subsec:graph-hpo})
or $\rho:R\to L$ 
(as in sections~\ref{subsec:graph-dpo} and~\ref{subsec:graph-sqpo}).
A morphism $\pi:\rho\to \rho_1$ in $\catP$, 
with $\funL(\pi)=f:L\to L_1$ and $\funR(\pi)=g:R\to R_1$,
is illustrated as follows:
$$ \xymatrix@C=2pc@R=1.5pc{
L \ar[d]_{f} \ar@{~>}[rr]^{\rho} & \ar@{}[d]|{\pi} & R \ar[d]^{g} \\ 
L_1 \ar@{~>}[rr]_{\rho_1} && R_1 \\ 
}$$
Then, each rewriting process $\funS_\rho$ can be illustrated as:
$$ \xymatrix@C=2pc@R=1.5pc{
L \ar[d]_{f} \ar@{~>}[rr]^{\rho} & & R \\ 
L_1 &&  \\ 
} \quad
\xymatrix@R=.8pc{ \\ \ar@{|->}[r]^{\funS_\rho} & \\}  \quad
\xymatrix@C=2pc@R=1.5pc{
L \ar[d]_{f} \ar@{~>}[rr]^{\rho} & \ar@{}[d]|{\funS_\rho(f)} & R \ar[d]^{g} \\ 
L_1 \ar@{~>}[rr]_{\rho_1} && R_1 \\ 
}$$

For instance, definition~\ref{definition:crs-po} below provides 
categorical rewriting systems based on pushouts. 
As usual a category \emph{with pushouts} is a category $\catC$ 
such that for every morphisms $f$ and $\rho$ in $\catC$ with the same source,  
the pushout of $\rho$ and $f$ exists in $\catC$. 
The \emph{category of arrows} of any category $\catC$ is denoted $\catC^{\to}$:
its objects are the morphisms of $\catC$ 
and its morphisms are the commutative squares in $\catC$. 

\begin{definition}
\label{definition:crs-po}
Let $\catC$ be a category with pushouts. 
The \emph{categorical rewriting system based on pushouts in $\catC$},
denoted as $\rs_{\po,\catC}$, is made of 
the categories $\catL=\catR=\catC$ and $\catP=\catC^{\to}$, 
the source functor $\funL=\funSrc:\catC^{\to}\to\catC$,  
the target functor $\funR=\funTgt:\catC^{\to}\to\catC$, 
and the family of functions $\funS_{\po}$ such that
for each rule $\rho$ the function $\funS_{\po,\rho}$ is total and 
for each match $f$ the commutative square 
$\funS_{\po,\rho}(f)$ is defined as the pushout of $\rho$ and $f$ in $\catC$. 
$$ \xymatrix@C=2pc@R=1.5pc{
L \ar[d]_{f} \ar[rr]^{\rho} & & R \\ 
L_1 &&  \\ 
} \quad
\xymatrix@R=.8pc{ \\ \ar@{|->}[r]^{\funS_{\po,\rho}} & \\}  \quad
\xymatrix@C=2pc@R=1.5pc{
L \ar[d]_{f} \ar[rr]^{\rho} & \ar@{}[d]|{\funS_{\po,\rho}(f)} & R \ar[d]^{g} \\ 
L_1 \ar[rr]_{\rho_1} && R_1 
    \ar@{-}[]+L+<-6pt,+1pt>;[]+LU+<-6pt,+6pt> 
    \ar@{-}[]+U+<-1pt,+6pt>;[]+LU+<-6pt,+6pt> 
\\ 
}$$
\end{definition}

In section~\ref{sec:graph}, we consider categorical rewriting systems 
which generalize the pushout rewriting systems. 
There is a need for these generalizations, 
since there may be restrictions (e.g., injectivity conditions
or gluing conditions) 
on the morphisms used for rules and for matches. 
These generalizations are built according to the following patterns. 

\begin{definition}
\label{definition:crs-arrows}
Let $\catC$ be a category with two wide subcategories $\catM$ and $\catD$. 
The \emph{generalized arrow category} $\catD^{\to\catM}$ (in $\catC$)
is the following category:
the objects in 
$\catD^{\to\catM}$ are the morphisms in $\catD$, and the morphisms
from $\rho$ to $\rho_1$ in $\catD^{\to\catM}$, where $\rho:L\to R$ and
$\rho_1:L_1\to R_1$ in $\catD$, are the pairs $(f:L\to L_1, g:R\to
R_1)$ of morphisms in $\catM$ such that $g\circ\rho = \rho_1\circ f$ in $\catC$.
The \emph{source functor} $\funSrc:\catD^{\to\catM}\to\catM$ and the
\emph{target functor} $\funTgt:\catD^{\to\catM}\to\catM$ map each
object $\rho$ in $\catD^{\to\catM}$ to its source and target,
when $\rho$ is seen as a morphism in $\catD$;
they map each morphism $(f,g)$ in $\catD^{\to\catM}$ 
to the morphisms $f$ and $g$ in $\catM$, respectively.
\end{definition}

This situation yields two spans of categories where
$\catL=\catR=\catM$ and $\catP=\catD^{\to\catM}$, 
as defined below; these spans will be used for 
describing graph transformation systems as categorical rewriting systems 
in sections~\ref{sec:graph} and~\ref{sec:garbage}.

\begin{definition} 
\label{definition:crs-arrows-span} 
Let $\catC$ be a category with two wide subcategories $\catM$ and $\catD$. 
Let $\catD^{\to\catM}$ denote the corresponding generalized arrow category 
and $\funSrc,\funTgt:\catD^{\to\catM}\to\catM$ the source and target functors.
\begin{itemize}
\item The \emph{direct arrows-based span on $\catC$ with rules in $\catD$ 
and matches in $\catM$} is the span of categories 
$(\spa{\catM}{\funSrc}{\catD^{\to\catM}}{\funTgt}{\catM})$.
This means that 
a rule $\rho:L\arsto R$ is a morphism $\rho:L\mo R$ in $\catD$, 
a match is a morphism in $\catM$ and 
a morphism of rules (from $\rho$ to $\rho_1$) is a commutative square in $\catC$
with $f,g$ in $\catM$:
$$ \xymatrix@C=2pc@R=1.5pc{
L \ar[d]_{f} \ar[rr]^{\rho} & \ar@{}[d]|{=} & R \ar[d]^{g} \\ 
L_1 \ar[rr]_{\rho_1} && R_1 \\ 
} $$
\item The \emph{inverse arrows-based span on $\catC$ with rules in $\catD$ 
and matches in $\catM$} is the span of categories 
$(\spa{\catM}{\funTgt}{\catD^{\to\catM}}{\funSrc}{\catM})$.
This means that 
a rule $\rho:L\arsto R$ is a morphism $\rho:R\mo L$ in $\catD$, 
a match is a morphism in $\catM$ and 
a morphism of rules (from $\rho$ to $\rho_1$) is a commutative square in $\catC$
with $f,g$ in $\catM$:
$$\xymatrix@C=2pc@R=1.5pc{
L \ar[d]_{f} & \ar@{}[d]|{=} & R \ar[d]^{g} \ar[ll]_{\rho} \\ 
L_1 && R_1 \ar[ll]^{\rho_1} \\ 
}$$ 
\end{itemize}
\end{definition}

\begin{remark} 
\label{remark:crs-coslice}  
For any category $\catC$ and any object $X$ in $\catC$,
let $X\coslice\catC$ denote the \emph{coslice category} 
of objects of $\catC$ under $X$.
Then the objects of $X\coslice\catC$ are 
the morphisms in $\catC$ with source $X$.
Let $\rs= (\spa{\catL}{\funL}{\catP}{\funR}{\catR},\funS)$,
be a categorical rewriting system.  
For each rule $\rho:L\arsto R$ let 
$\funL_\rho:\rho\coslice\catP \to L\coslice\catL$
denote the functor induced by $\funL$. 
Then $\funS_\rho$ can be seen as a partial function 
$ \funS_\rho:|L\coslice\catL| \parmo |\rho \coslice\catP| $ 
such that $\funL_\rho\circ\funS_\rho$ is the identity of $\dom(\funS_\rho)$.
\end{remark} 

\subsection{Functoriality of categorical rewriting systems} 
\label{subsec:crs-fun}

A categorical rewriting system, when it is seen as an abstract rewriting system, 
is read ``horizontally'': it maps the left-hand side match $f:L\to L_1$ 
to the right-hand side match $g:R\to R_1$.
But it may also be read ``vertically'': it maps the rule $\rho:L\arsto R$ 
to the rule $\rho_1:L_1\arsto R_1$.
In this section we study a functoriality property of categorical rewriting systems 
from this ``vertical'' point of view;
a similar property is called ``vertical composition'' in \cite{Lowe10}. 
The statements and results below are given up to isomorphism.

\begin{definition}
\label{definition:crs-fun} 
A categorical rewriting system 
$(\spa{\catL}{\funL}{\catP}{\funR}{\catR},\funS)$ 
is \emph{functorial} if for each rule $\rho:L\arsto R$ the partial function $\funS_\rho$ 
satisfies:
\begin{itemize}
\item the identity $\id_L$ is in the domain of $\funS_\rho$ and 
  $$ \funS_\rho(\id_L)=\id_\rho \;. $$
$$ \xymatrix@C=2pc@R=1.5pc{
L \ar[d]_{\id_L} \ar@{~>}[rr]^{\rho} & & R \\ 
L &&  \\ 
} 
\quad\xymatrix@R=.8pc{ \\ \ar@{|->}[r]^{\funS_\rho} & \\} \quad
\xymatrix@C=2pc@R=1.5pc{
L \ar[d]_{\id_L} \ar@{~>}[rr]^{\rho} & \ar@{}[d]|{\id_\rho} & R \ar[d]^{\id_R} \\ 
L \ar@{~>}[rr]_{\rho} && R \\ 
}
$$
\item and for each pair of consecutive morphisms $f_1:L\to L_1$ 
and $f_2:L_1\to L_2$ in $\catL$, 
if $f_1\in\dom(\funS_\rho)$ and $f_2\in\dom(\funS_{\rho_1})$,  
where $\rho_1$ denotes the target of $\funS_\rho(f_1)\,$,
then $f_2\circ f_1\in\dom(\funS_\rho)$ and 
  $$ \funS_{\rho_1}(f_2)\circ \funS_\rho(f_1) = \funS_\rho(f_2\circ f_1) \;. $$
$$  \xymatrix@C=2pc@R=1.5pc{
L \ar[d]_{f_1} \ar@{~>}[rr]^{\rho} & & R \\ 
L_1 \ar[d]_{f_2} && \\ 
L_2 &&  \\ 
}
\quad\xymatrix@R=1.5pc{ \\ \ar@{|->}[r]^{\funS_\rho} & \\}  \quad
\xymatrix@C=2pc@R=1.5pc{
L \ar[d]_{f_1} \ar@{~>}[rr]^{\rho} & \ar@{}[d]|{\funS_\rho(f_1)} & R \ar[d]^{g_1} \\ 
L_1 \ar[d]_{f_2} \ar@{~>}[rr]|{\rho_1} & \ar@{}[d]|{\funS_{\rho_1}(f_2)} & R_1 \ar[d]^{g_2} \\ 
L_2 \ar@{~>}[rr]_{\rho_2} && R_2 \\ 
}  
\; \xymatrix@R=1.5pc{ \\ = \\}  \; 
\xymatrix@C=2pc@R=1.5pc{
L \ar[dd]|(.6){f_2\circ f_1} \ar@{~>}[rr]^{\rho} & \ar@{}[dd]|{\funS_\rho(f_2\circ f_1)} & 
  R \ar[dd]|(.6){g_2\circ g_1} \\ 
&\mbox{ \phantom{L}} & \\ 
L_2 \ar@{~>}[rr]_{\rho_2} && R_2 \\ 
}
$$
\end{itemize}
\end{definition}

For instance, using definition~\ref{definition:crs-po}, 
the next result is due to the well-known compositionality property of pushouts.

\begin{proposition}
\label{proposition:crs-po} 
Let $\catC$ be a category with pushouts.
The categorical rewriting system $\rs_{\po,\catC}$ is functorial. 
\end{proposition}

\begin{remark} 
\label{remark:crs-coslice-functorial}  
The name ``functorial'' comes from the interpretation of 
categorical rewriting systems in terms of coslice categories, 
as in remark~\ref{remark:crs-coslice}~: 
let $\rs= (\spa{\catL}{\funL}{\catP}{\funR}{\catR},\funS)$ 
be a categorical rewriting system,
and let us assume that for each rule $\rho:L\arsto R$ 
the rewriting process $\funS_\rho$ is total, 
which means that it is a total function 
$ \funS_\rho:|L\coslice\catL| \parmo |\rho \coslice\catP| $ 
such that $\funL_\rho\circ\funS_\rho$ is the identity of $|L\coslice\catL|$.
For each morphism $h:f_1\to f_2$ in $L\coslice\catL$,
i.e., for each morphism $h:L_1\to L_2$ in $\catL$ such that $h\circ f_1 = f_2$, 
let us define $\funS_\rho(h:f_1\to f_2)=\funS_{\rho_1}(h)$  
where $\rho_1$ is the target of $\funS_\rho(f_1)$ in $\catP$. 
Then it can be proved that 
$\rs$ is functorial if and only if for each rule $\rho:L\arsto R$,
$ \funS_\rho(\id_L)=\id_\rho$ and $\funS_\rho$ is a functor 
$ \funS_\rho:L\coslice\catL\to \rho\coslice\catP $.
\end{remark}

\subsection{Composition of categorical rewriting systems} 
\label{subsec:crs-comp}

In order to compose (``horizontally'') categorical rewriting systems, 
we use composition of spans:
given two spans of categories 
$\spa{\catL}{\funL}{\catP}{\funR}{\catR}$
and $\spa{\catL'}{\funL'}{\catP'}{\funR'}{\catR'}$ 
which are consecutive, in the sense that $\catR=\catL'$,
the composed span
$\spa{\catL}{\funL''}{\catP''}{\funR''}{\catR'}$
is obtained from the pullback of $\funR$ and $\funL'$,
as follows:
  $$ 
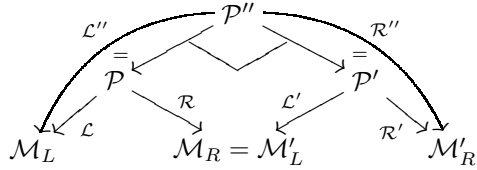
\xymatrix@=1pc{
  && \ar@/_4ex/[lldd]_{\funL''}^{=} \ar[ld] \catP'' \ar[rd] 
  \ar@/^4ex/[rrdd]^{\funR''}_{=}  
  \ar@{-}[]+R+<+9pt,-11pt>;[]+D+<+0pt,-15pt> 
                           \ar@{-}[]+L+<-9pt,-11pt>;[]+D+<+0pt,-15pt>
  &&  \\ 
  & \ar[ld]^{\funL} \catP \ar[rd]^{\funR} & 
  & \ar[ld]_{\funL'} \catP' \ar[rd]_{\funR'} & \\
  \catL & & \catR = \catL' & & \catR' \\ 
  } $$
The objects of $\catP''$ are the pairs $(\rho,\rho')$ with 
$\rho$ in $\catP$ and $\rho'$ in $\catP'$ such that 
$\funR(\rho)=\funL'(\rho')$. 
The morphisms from $\rho''=(\rho,\rho')$ to $\rho''_1=(\rho_1,\rho'_1)$ 
in $\catP''$ are the pairs $\pi''=(\pi,\pi')$ where 
$\pi:\rho\to\rho_1$ in $\catP$ and $\pi':\rho'\to\rho'_1$ in $\catP'$ 
are such that $\funR(\pi)=\funL'(\pi')$. 

\begin{definition}
\label{definition:crs-comp}
Let $\rs=(\spa{\catL}{\funL}{\catP}{\funR}{\catR},\funS)$
and $\rs'=(\spa{\catL'}{\funL'}{\catP'}{\funR'}{\catR'},\funS')$
be two categorical rewriting systems
which are consecutive, in the sense that $\catR=\catL'$. 
The \emph{composition} of $\rs$ and $\rs'$ is the categorical 
rewriting system  
$$\rs'\circ\rs=(\spa{\catL}{\funL''}{\catP''}{\funR''}{\catR'},
\funS''_{(\rho,\rho')}) $$ 
where $\spa{\catL}{\funL''}{\catP''}{\funR''}{\catR'}$ is the composition
of the spans in $\rs$ and $\rs'$ and where the family of partial functions 
$\funS''=(\funS''_{\rho''})_{\rho''\in|\catP''|}$ is defined as follows, 
for each $\rho''=(\rho,\rho')$ in $\catP''$:
the domain of $\funS''_{\rho''}$ is made of the 
morphisms $f$ in $\dom(\funS_\rho)$ such that 
$\funR(\funS_\rho(f))$ is in $\dom(\funS'_{\rho'})$, and for
each $f\in\dom(\funS''_{\rho''})$:   
   $$ \funS''_{(\rho,\rho')}(f) = (\funS_\rho(f), \funS'_{\rho'}(f')) 
   \; \mbox{ where }\; f'= \funR(\funS_\rho(f)) \;.$$
$$ 
\xymatrix@C=2pc@R=1.5pc{
L \ar[d]_{f} \ar@{~>}[rr]^{\rho} & \ar@{}[d]|{\funS_\rho(f)} & 
R=L' \ar[d]^{f'} \ar@{~>}[rr]^{\rho'} & \ar@{}[d]|{\funS_{\rho'}(f')} & 
R' \ar[d]^{f''}\\ 
L_1 \ar@{~>}[rr]_{\rho_1} && 
R_1=L'_1 \ar@{~>}[rr]_{\rho'_1} && 
R'_1 \\ 
}  \quad
\xymatrix@R=.8pc{  \\ = \\}  \quad
 \xymatrix@C=2pc@R=1.5pc{
L \ar[d]_{f} \ar@{~>}[rr]^{\rho''} & \ar@{}[d]|{\funS_{(\rho,\rho')}(f)} & R' \ar[d]^{f''} \\ 
L_1 \ar@{~>}[rr]_{\rho''_1} && R'_1 \\ 
}$$
\end{definition}

This composition gives rise to the bicategory of categorical rewriting systems
(as for spans, we get a bicategory rather than a category,
because the unicity of pushouts is only up to isomorphim).
The next result follows easily from the definitions.

\begin{proposition}
\label{proposition:crs-comp}
Let $\rs$ and $\rs'$ be two consecutive categorical rewriting systems.
If $\rs$ and $\rs'$ are functorial then $\rs'\circ\rs$ is functorial. 
\end{proposition}

\section{Functoriality of graph transformations}
\label{sec:graph}

Following \cite{EhrigPS73} a lot of graph transformation systems 
have been studied in an algebraic approach. 
We show that many of them can be seen as categorical rewriting systems
which satisfy the functoriality property.  
A direct arrows-based span  
is used in sections \ref{subsec:graph-spo} and \ref{subsec:graph-hpo} 
for single pushout and heterogeneous pushout rewriting systems.
In sections \ref{subsec:graph-dpo} and \ref{subsec:graph-sqpo},
for double pushout and sesqui-pushout rewriting systems, 
an inverse arrows-based span is used, then a direct one, 
and finally both are composed according to definition~\ref{definition:crs-comp}.  
We define a \emph{graph} as a set of \emph{nodes} and a set of \emph{edges} 
with two functions from edges to nodes called the \emph{source} 
and the \emph{target} functions. 
A \emph{morphism of graphs} is made of a function on nodes and a function on edges 
which preserve the sources and targets. 
This provides the category of graphs, denoted as $\catgraph$. 

\subsection{Single Pushout rewriting}
\label{subsec:graph-spo}

In this section we show that, under suitable assumptions, 
the single pushout approach to graph transformation (SPO) \cite{EhrigHKLRWC97}
can be seen as a categorical rewriting system.  
Let $\catM_\spo=\catgraph$ be the category of graphs.  
Let $\catC_\spo=\catgraphp$ be the category of graphs with partial morphisms, 
so that $\catM_\spo$ can be seen as a wide subcategory of $\catC_\spo$.
Let $\catD_\spo=\catgraphpm$ be the wide subcategory of $\catC_\spo$
with partial monomorphisms.
We consider the direct arrows-based span on $\catC_\spo$ 
with rules in $\catD_\spo$ and matches in $\catM_\spo$. 
Following \cite[Definition 7]{EhrigHKLRWC97}, 
given a rule $r:L\monoparmo R$, 
we say that a match $f:L\mo L_1$ is \emph{conflict-free} with respect to $r$ 
when $f$ does not identify any item (node or edge) in the domain of $r$ 
with an item outside this domain.
For each rule $r:L\monoparmo R$, 
we define $\funS_{\spo,r}$ as the partial function 
with domain the conflict-free matches with respect to $r$, 
such that $\funS_{\spo,r}(f)$ is the pushout of $f$ and $r$ in $\catgraphp$ 
for each $f$ in $\dom(\funS_{\spo,r})$. 
It follows from \cite[Proposition 5 and Lemma 8]{EhrigHKLRWC97} 
that this pushout exists, that $r_1$ is a partial monomorphism 
and that $g$ is a total morphism.  
$$ \xymatrix@C=2pc@R=1.5pc{
L \ar[d]_{f} \ar@{>-^>}[rr]^{r} & & R \\ 
L_1 &&  \\ 
} \quad
\xymatrix@R=.8pc{ \\ \ar@{|->}[r]^{\funS_{\spo,r}} & \\}  \quad
\xymatrix@C=2pc@R=1.5pc{
L \ar[d]_{f} \ar@{>-^>}[rr]^{r} & \ar@{}[d]|{\funS_{\spo,r}(f)} & R \ar[d]^{g} \\ 
L_1 \ar@{>-^>}[rr]_{r_1} && R_1 
    \ar@{-}[]+L+<-6pt,+1pt>;[]+LU+<-6pt,+6pt> 
    \ar@{-}[]+U+<-1pt,+6pt>;[]+LU+<-6pt,+6pt> 
\\ 
}$$ 

\begin{definition}
\label{definition:graph-spo}
The \emph{categorical rewriting system for graphs based on single pushouts},
denoted as $\rs_\spo$, 
is made of the direct arrows-based span on $\catC_\spo=\catgraphp$ 
with rules in $\catD_\spo=\catgraph$ 
and matches in $\catM_\spo=\catgraphpm$
together with the family of partial functions $\funS_\spo$
defined as above from pushouts in $\catgraphp$.
\end{definition}

\begin{lemma}
\label{lemma:graph-spo} 
Let us consider the categorical rewriting system $\rs_\spo$. 
Let $r:L \monoparmo R$ be a rule and $f_1:L \mo L_1$ a match 
which is conflict-free with respect to $r$. 
Let $R_1$ with $r_1:L_1 \monoparmo R_1$ and $g_1:R\mo R_1$ 
be the pushout of $r$ and $f_1$ in $\catgraphp$. 
Let $f_2:L_1 \mo L_2$ be a match 
which is conflict-free with respect to $r_1$. 
Then $f_2 \circ f_1$ is conflict-free with respect to $r$. 
\end{lemma}

\begin{proof}
Let $f=f_2 \circ f_1:L \mo L_2$. 
The proof is done by contradiction. 
Let us assume that there are two items $x$ and $y$ in $L$
such that $f(x)=f(y)$, with $x\in\dom(r)$ and $y\not\in\dom(r)$. 
Then there are two cases:
  \begin{enumerate} 
  \item If $f_1(x)=f_1(y)$ 
    then $f_1$ is not conflict-free with respect to $r$.
  \item Otherwise let $x_1=f_1(x)$ and $y_1=f_1(y)$, so that $f_2(x_1)=f_2(y_1)$. 
    The commutativity of the square $\funS_{\spo,r}(f_1)$ 
    is written as $g_1\circ r = r_1\circ f_1 $. 
    This implies that $g_1\circ r $ and $r_1\circ f_1 $ have the same domain,
    and since $f_1$ and $g_1$ are total  
    this means that for each item $x$ in $L$, 
    $x\in\dom(r)$ if and only if $f_1(x)\in\dom(r_1)$.
    Thus, $x_1\in\dom(r_1)$ and $y_1\not\in\dom(r_1)$,
    so that $f_2$ is not conflict-free with respect to $r_1$.
  \end{enumerate}
\end{proof}

\begin{proposition}
\label{proposition:graph-spo} 
The categorical rewriting system $\rs_\spo$ is functorial.
\end{proposition}

\begin{proof}
This is due to lemma~\ref{lemma:graph-spo}
and to the well-known compositionality property of pushouts. 
\end{proof}

\subsection{Heterogeneous pushout rewriting}
\label{subsec:graph-hpo}

We now consider the heterogeneous pushout framework (HPO) presented in 
\cite{DuvalEP09}, which allows some deletion and cloning 
in the context of termgraph rewriting.
Given a set called the set of \emph{labels}, 
with an \emph{arity} (a natural number) for each label,
a \emph{termgraph} is a graph where some nodes are labeled, 
when a node $n$ has a label $\ell$  
then the successors of $n$ form a totally ordered set and 
their number is the arity of $\ell$, 
and when a node $n$ is unlabeled then it has no successor. 
If $G$ is a termgraph then $|G|$ denotes the set of nodes of $G$. 
A \emph{morphism of termgraphs} 
(respectively a \emph{partial morphism of termgraphs}) 
is a morphism of graphs
(respectively a partial morphism of graphs) 
which maps labeled nodes to labeled nodes,
preserving the labels and the ordering of the successors.  This provides 
the category of termgraphs $\cattermgraph$.
Let $\catM_\hpo=\cattermgraphm$ be the wide subcategory of $\cattermgraphm$
with monomorphisms.  Let $\catC_\hpo$ be the category
with the termgraphs as objects and with morphisms from $L$ to $R$ the
pairs $(\tau,\sigma)$ 
of partial termgraph morphisms $\tau:L\parmo R$ and $\sigma:R\parmo L$.  
Then $\catM_\hpo$
is considered as a wide subcategory of $\catC_\hpo$ by identifying
each total morphism of termgraphs $f:L\mo L_1$ to the pair
$(f,\omega)$ where $\omega:L_1\parmo L$ is nowhere defined. 
Let $\catD_\hpo$ be the wide subcategory of $\catC_\hpo$ with
morphisms the pairs $\rho=(\tau,\sigma):L\hetmo R$ such that the
domain of $\tau$ is the set of nodes of $L$ and the domain of $\sigma$
is a subset of the set of nodes of $R$. Moreover, every 
node $p \in |R|$ in the domain of $\sigma$ is either unlabelled or 
such that the node $q=\sigma(p) \in |L|$ is such that 
$p$ and $q$ share the same label and 
the successors of $p$ in $R$ are the image 
by $\tau$ of the successors of $q$ in $L$.  

We consider the direct arrows-based span on
$\catC_\hpo$ with rules in $\catD_\hpo$ and matches in $\catM_\hpo$.
Following \cite[Definitions~6 and~7]{DuvalEP09}, 
for each rule $\rho:L\hetmo R$ and each match $f:L\monomo L_1$, 
a \emph{heterogeneous cocone over $\rho$ and $f$} is 
made of a rule $\rho_1:L_1\hetmo R_1$ and a match $g:R\monomo R_1$ 
such that $\rho_1\circ f = g\circ \rho$ in $\catC_\hpo$. 
A \emph{morphism of heterogeneous cocones over $\rho$ and $f$}, 
say $h:(\rho_1,g)\to(\rho_1',g')$,  
is a morphism $h:R_1\mo R_1'$ in $\catM_\hpo$ such that 
$h\circ\rho_1=\rho'_1$ and $h\circ g=g'$ in $\catC_\hpo$.
This yields the category of heterogeneous cocones over $\rho$ and $f$,
and a \emph{heterogeneous pushout of $\rho$ and $f$} 
is defined as an initial object in this category.
The unicity of the heterogeneous pushout, up to isomorphism, 
is a consequence of its initiality property. 
Its existence is proven in \cite[theorem~1]{DuvalEP09} 
by providing an explicit construction. 
For each rule $\rho:L\hetmo R$  
let us define $\funS_{\hpo,\rho}$ as the total function 
such that $\funS_{\hpo,\rho}(f)$ is the heterogeneous pushout of $f$ 
and $\rho$ for each match $f$, which is denoted as:
$$  \xymatrix@C=2pc@R=1.5pc{
  L \ar@{^{<}-^{>}}[rr]^{\rho} \ar@{>->}[d]_{f} && 
      R  \\
   L_1 &  \\ }
\quad
\xymatrix@R=.8pc{ \\ \ar@{|->}[r]^{\funS_{\hpo,\rho}} & \\}  \quad
\xymatrix@C=2pc@R=1.5pc{
  L \ar@{^{<}-^{>}}[rr]^{\rho} \ar@{>->}[d]_{f} & 
    \ar@{}[d]|(.4){\funS_{\hpo,\rho}(f)} & R \ar@{>->}[d]^{g} \\
  L_1 \ar@{^{<}-^{>}}[rr]_{\rho_1} & & R_1 
    \ar@{--}[]+L+<-6pt,+1pt>;[]+LU+<-6pt,+6pt> 
    \ar@{--}[]+U+<-1pt,+6pt>;[]+LU+<-6pt,+6pt> 
\\ }$$
It follows from \cite[Proposition~1]{DuvalEP09} 
that this construction provides a rule $\rho_1$ and a match $g$,
so that we get a categorical rewriting system. 

\begin{definition}
\label{definition:graph-hpo}
The \emph{categorical rewriting system for termgraphs based on heterogeneous pushouts},
denoted as $\rs_\hpo$, 
is made of the direct arrows-based span on $\catC_\hpo$ 
with rules in $\catD_\hpo$ 
and matches in $\catM_\hpo=\cattermgraph$ 
together with the family of partial functions $\funS_\hpo$
defined as above from heterogeneous pushouts. 
\end{definition}

\begin{proposition}
\label{proposition:graph-hpo} 
The categorical rewriting system $\rs_\hpo$ is functorial.
\end{proposition}

\begin{proof} 
The compositionality property of heterogeneous pushouts, 
similar to the compositionality property of pushouts,
follows easily from their initiality property.
Proposition~\ref{proposition:graph-hpo} is a consequence of this property.
\end{proof}

\subsection{Double pushout rewriting} 
\label{subsec:graph-dpo}

In this section
we check that under suitable assumptions 
the graph transformation based on double pushouts (DPO) 
\cite{CorradiniMREHL97} can be considered as a categorical rewriting system 
which is composed, in the sense of definition~\ref{definition:crs-comp}, 
of a categorical rewriting system based on pushout complements 
(as defined below) 
followed by a categorical rewriting system based on pushouts
(definition~\ref{definition:crs-po}). 
We restrict our study to cases where the pushout complement is unique. 
Let $\catM_\poc=\catC_\poc=\catgraph$ be the category of graphs.  
Let $\catD_\poc=\catgraphm$ be the wide subcategory of $\catC_\poc$
with injective morphisms.
We consider the inverse arrows-based span on $\catC_\poc$ 
with rules in $\catD_\poc$ 
and matches in $\catM_\poc$.
This means that a rule $\rho:L\arsto R$ is a monomorphism of graphs 
$\rho:R\monomo L$, or (according to the usual notations) $l:K\monomo L$.
Given a graph $G$ and a subgraph $H$ of $G$,
we denote as $G-H$ the \emph{partial graph} 
made of the nodes and edges in $G$ which are not in $H$,
with the restriction of the source and target functions.
In general $G-H$ is not a graph, since it can have \emph{dangling} edges, 
i.e., edges which are not in $H$ but which have their source or target in $H$. 
Following \cite[Proposition 9]{CorradiniMREHL97}, 
given a rule $l:K\monomo L$ 
we say that a match $f:L\mo L_1$ satisfies the \emph{gluing condition} 
with respect to $l$ if:
\begin{itemize}
  \item \emph{Dangling condition.} 
     If an edge $e_1$ in $L_1$ is incident to a node in $f(L-l(K))$ 
     then $e_1$ is in $f(L)$.
  \item \emph{Identification condition.} 
     If two nodes (respectively two edges) $x$ and $y$ in $L$ 
     are such that $x \ne y$ and $f(x)=f(y)$ 
     then $x$ and $y$ are in $l(K)$.
\end{itemize}
One can remark that if the dangling condition is satisfied then 
$L_1-f(L-l(K))$ is a graph. 
It is proven in \cite[Proposition 9]{CorradiniMREHL97} 
that when $f$ satisfies the gluing condition with respect to $l$
then the graph $K_1 = L_1-f(L-l(K))$ together with 
the inclusion $l_1:K_1\monomo L_1$ 
and the morphism $g:K\mo K_1$ which maps each node or edge $x$ to 
$f(l(x))$ 
forms a pushout complement of $l$ and $f$ in $\catgraph$, 
and in addition this pushout complement is unique up to isomorphism. 
For each rule $l:K\monomo L$  
we define $\funS_{\poc,l}$ as the partial function 
with domain the matches with source $L$ 
which satisfy the gluing condition with respect to $l$,  
such that $\funS_{\poc,\rho}(f)$ is the pushout complement of $l$ and $f$  
for each $f$ in $dom(\funS_{\poc,l})$: 
$$ \xymatrix@C=2pc@R=1.5pc{
L \ar[d]_{f} & & K \ar@{>->}[ll]_{l} \\ 
L_1 &&  \\ 
} \quad
\xymatrix@R=.8pc{ \\ \ar@{|->}[r]^{\funS_{\poc,\rho}} & \\}  \quad
\xymatrix@C=2pc@R=1.5pc{
L \ar[d]_{f} & \ar@{}[d]|{\funS_{\poc,\rho}(f)} & K \ar@{>->}[ll]_{l} \ar[d]^{g} \\ 
L_1 
    \ar@{-}[]+R+<+6pt,+1pt>;[]+RU+<+6pt,+6pt> 
    \ar@{-}[]+U+<+1pt,+6pt>;[]+RU+<+6pt,+6pt> 
&& K_1 \ar@{>->}[ll]_{l_1} \\ 
}$$ 

\begin{definition}
\label{definition:graph-dpo}
The \emph{categorical rewriting system for graphs based on pushout complements},
denoted as $\rs_\poc$, 
is made of the inverse arrows-based span on $\catC_\poc=\catgraph$ 
with rules in $\catD_\poc=\catgraphm$ 
and matches in $\catM_\poc=\catgraph$
together with the family of partial functions $\funS_\poc$
defined as above from pushout complements in $\catgraph$.
The \emph{categorical rewriting system for graphs based on double pushouts},
denoted as $\rs_\dpo$, 
is the composition of $\rs_\poc$ and $\rs_{\po,\catgraph}$
(from definition~\ref{definition:crs-po}). 
\end{definition}

\begin{lemma}
\label{lemma:graph-dpo} 
Let us consider the categorical rewriting system $\rs_\poc$. 
Let $l:K \monomo L$ be a rule and $f_1:L \mo L_1$ a match  
which satisfies the gluing condition with respect to $l$. 
Let $(K_1,l_1,g_1)$ be the pushout complement of $l$ and $f_1$. 
Let $f_2:L_1 \to L_2$ be a match  
which satisfies the gluing condition with respect to $l_1$. 
Then $f_2 \circ f_1$ satisfies the gluing condition with respect to $l$. 
\end{lemma}

\begin{proof}
Let $f=f_2 \circ f_1:L \to L_2$
We have to prove that $f$ satisfies the
dangling condition and the identification conditions with respect to $l$. 
  \begin{itemize}
    \item \emph{Dangling condition.}  
    Suppose that $f_1$ and $f_2$ verify the identification condition.
    Let $e_2$ be an edge in $L_2$ which is incident to 
    a node $x_2$ in $f(L - l(K))$. 
    We have to prove that $e_2$ is in $f_2(f_1(L))=f(L)$.
    There are two cases:
    \begin{enumerate}
      \item There exists an edge $e_1$ in $L_1$ such that
            $e_2=f_2(e_1)$. Let $x$ be a node in $L-l(K)$ such that 
            $x_2=f(x)$, and $x_1=f_1(x)$. We know that $e_1$ is incident to 
            $x_1$ since $f_2(e_1)$ is incident to $x_2$, indeed if it were 
            not the case then $f_2(x_1)=f_2(z)=x_2$ with $z \not = x_1$ and 
            the identification condition of $f_2$ would be violated because 
            $z,x_1$ are not in $l_1(K_1)$. Moreover, 
            since $f_1$ satisfies the 
            dangling condition with respect to $l$ then $e_1$
            is in $f_1(L)$, thus $f_2(e_1)=e_2$ is in $f_2(f_1(L))=f(L)$. 
      \item The edge $e_2$ has no $f_2$-antecedent in $L_1$.  Let $x$
        be a node of $L-l(K)$ such that $f(x)=x_2$. Let $x_1=f_1(x)$,
        then $x_{1} \in L_1-l_1(K_1)$ because let $(K_1,l_1,g_1)$ be 
        the pushout complement of $l$ and $f_1$, it is unique and $K_1$ is
        the subgraph of $L$ obtained by removing all items that are in
        the image of $f_1$ but not in the image of $f_1 \circ l$ (see
        \cite[Proposition 9]{CorradiniMREHL97}). Thus $e_2$ is an edge
        incident to a node of $f_2(L_1-l_1(K_1))$.  Since $f_2$
        satisfies the dangling condition with respect to $l_1$, we
        know that $e_2$ is in $f_2(L_1)$, which contradicts our
        hypothesis that $e_2$ has no $f_2$-antecedent. Thus, this case
        cannot occur.
    \end{enumerate}
    \item \emph{Identification condition.} 
    Suppose that there are two items $x,y \in L$ 
    such that $x \ne y$ and $f(x) = f(y)$.
    We have to prove that $x$ and $y$ are in $l(K)$. 
    Then there are two cases:
       \begin{enumerate}
        \item If $f_1(x) = f_1(y)$, the identification condition 
          of $f_1$ with respect to $l$ implies that $x$ and $y$ are in $l(K)$. 
        \item If $f_1(x) \ne f_1(y)$, let $x_1=f_1(x)$ and $y_1=f_1(y)$,
          so that $x_1 \ne y_1$ and $f_2(x_1) = f_2(y_1)$. 
          The identification condition 
          of $f_2$ with respect to $l_1$ implies that 
          $x_1$ and $y_1$ are in $l_1(K_1)$. Now since 
          $K_1=L_1-f_1(L-l(K))$ and $x_1, y_1$ are in $f_1(L)$, it 
          implies that  they are in $l(K)$. 
      \end{enumerate}
  \end{itemize}
\end{proof}

\begin{proposition}
\label{proposition:graph-dpo} 
The categorical rewriting systems $\rs_\poc$ and $\rs_\dpo$ are functorial.
\end{proposition}

\begin{proof} 
  The functoriality of $\rs_\poc$ follows from lemma~\ref{lemma:graph-dpo}
  and the compositionality property of pushouts.
  Then the functoriality of $\rs_\dpo$ follows from 
  the functoriality of $\rs_{\po,\catgraph}$
  (proposition~\ref{proposition:crs-po}) 
  and from proposition~\ref{proposition:crs-comp}. 
\end{proof} 
  
\subsection{Sesqui-pushout rewriting}
\label{subsec:graph-sqpo}

Similarly to section~\ref{subsec:graph-dpo}, 
under suitable assumptions 
the graph transformation based on sesqui-pushouts (SqPO) 
\cite{CorradiniHHK06} can be considered as a categorical rewriting system 
which is composed 
of a categorical rewriting system based on final pullback complements 
(as defined below) 
followed by a categorical rewriting system based on pushouts. 
Final pullback complements are defined in \cite[Theorem 4.4]{DyckThol87} 
as follows.  
For each match $f:L\to L_1$ let us consider the slice categories 
$\catD\slice L$ and $\catD\slice L_1$ of objects of $\catD$ 
over $L$ and $L_1$, respectively.
Let $f^{*}:\catD\slice L_1 \to \catD\slice L$ denote the pullback functor, 
which maps each $l_1:K_1\to L_1$ to $f^{*}(l_1):K \to L$
such that there is a pullback square: 
$$ \xymatrix@C=2pc@R=1.5pc{
L \ar[d]_{f} & & K \ar[d] \ar[ll]_{f^{*}(l_1)}
    \ar@{-}[]+L+<-6pt,-1pt>;[]+LD+<-6pt,-6pt> 
    \ar@{-}[]+D+<-1pt,-6pt>;[]+LD+<-6pt,-6pt>  
 \\ 
L_1 && K_1  \ar[ll]^{l_1} \\ 
}$$ 
The \emph{Dyckhoff-Tholen condition} for $f$ states that 
the pullback functor $f^{*}$ has a right adjoint 
$f_{*}$ such that $f^{*} \circ f_{*}$ is the identity.
This last condition implies that 
the functor $f_{*}:\catD\slice L \to \catD\slice L_1$
provides a pullback complement for $f$ and $l$, for every $l:K \to L$, 
which is called the \emph{final pullback complement} (FPBC) of $f$ and $l$. 
The definition of the final pullback complement of $f$ and $l$ implies
that, when it does exist, it is unique. 
Let $\catC_\fpbc=\catgraph$ be the category of graphs, 
and let $\catgraphm$ be the category of graphs with monomorphisms, 
seen as a wide subcategory of $\catgraph$. 
Following \cite{CorradiniHHK06}, we define two kinds of 
rewriting systems based on FPBCs. 
In the first one the rules are monomorphisms,
in the second one the matches are monomorphisms.
In both cases we consider an inverse arrows-based span on $\catgraph$.
\begin{enumerate}
\item \emph{Left-linear rules.} 
Let $\catD_{\fpbc,1}=\catgraphm$ 
and $\catM_{\fpbc,1}=\catgraph$. 
Following \cite[definition~4]{CorradiniHHK06}, 
given a rule $l:K\monomo L$  
we say that a match $f:L\mo L_1$ is \emph{conflict-free} with respect to $l$ 
when $f$ does not identify any item in the image of $l$ 
with an item outside this image 
(note the similarity with the definition of conflict-free matches for SPO). 
For each rule $l:K\monomo L$ 
we define $\funS_{\fpbc,1,l}$ as the partial function
with domain the conflict-free matches with respect to $l$,
such that $\funS_{\fpbc,1,l}(f)$ is the final pullback complement 
of $l$ and $f$ in $\catgraph$, for each $f$ in $\dom(\funS_{\fpbc,1,l})$. 
It is proved in \cite[construction 5]{CorradiniHHK06}
that this final pullback complement exists,
and that it yields $l_1:K_1\monomo L_1$ and $g:K\mo K_1$.  
$$ \xymatrix@C=2pc@R=1.5pc{
L \ar[d]_{f} & & K \ar@{>->}[ll]_{l} \\ 
L_1  &&  \\ 
} \quad
\xymatrix@R=.8pc{ \\ \ar@{|->}[r]^{\funS_{\fpbc,1,l}} & \\}  \quad
\xymatrix@C=2pc@R=1.5pc{
L \ar[d]_{f} & \ar@{}[d]|{\funS_{\fpbc,1,l}(f)} & 
  K \ar[d]^{g} \ar@{>->}[ll]_{l}
    \ar@{-}[]+L+<-6pt,-1pt>;[]+LD+<-6pt,-6pt> 
    \ar@{-}[]+D+<-1pt,-6pt>;[]+LD+<-6pt,-6pt>  
 \\ 
L_1
&& K_1  \ar@{>->}[ll]^{l_1} 
\\ 
}$$
\item \emph{Monic matches.} 
Let $\catD_{\fpbc,2}=\catgraph$ and $\catM_{\fpbc,2}=\catgraphm$. 
Given a rule $l:K\mo L$ 
we define $\funS_{\fpbc,2,l}$ as the total function on $\catgraphm$ 
such that $\funS_{\fpbc,2,l}(f)$ is the final pullback complement 
of $l$ and $f$ in $\catgraph$, for each $f$ in $\catgraphm$. 
It is proved in \cite[construction 6]{CorradiniHHK06} 
that this final pullback complement exists,
and that it yields $l_1:K_1\mo L_1$ and $g:K\monomo K_1$.  

$$ \xymatrix@C=2pc@R=1.5pc{
L \ar@{>->}[d]_{f} & & K \ar[ll]_{l} \\ 
L_1  &&  \\ 
} \quad
\xymatrix@R=.8pc{ \\ \ar@{|->}[r]^{\funS_{\fpbc,2,l}} & \\}  \quad
\xymatrix@C=2pc@R=1.5pc{
L \ar@{>->}[d]_{f} & \ar@{}[d]|{\funS_{\fpbc,2,l}(f)} & 
  K \ar@{>->}[d]^{g} \ar[ll]_{l}
    \ar@{-}[]+L+<-6pt,-1pt>;[]+LD+<-6pt,-6pt> 
    \ar@{-}[]+D+<-1pt,-6pt>;[]+LD+<-6pt,-6pt>  
 \\ 
L_1
&& K_1  \ar[ll]^{l_1} 
\\ 
}$$
\end{enumerate}

\begin{definition}
\label{definition:graph-sqpo}
The \emph{categorical rewriting systems for graphs based on final pullback complements},
denoted as $\rs_{\fpbc,i}$ with $i=1$ or $i=2$, 
are made of the inverse arrows-based span on $\catgraph$ 
with rules in $\catgraphm$ and matches in $\catgraph$ when $i=1$,
and with rules in $\catgraph$ and matches in $\catgraphm$ when $i=2$, 
together with the family of functions $\funS_{\fpbc,i}$
defined as above 
from final pullback complements in $\catgraph$,
so that $\funS_{\fpbc,1}$ is partial and $\funS_{\fpbc,2}$ is total. 
For each $i\in\{1,2\}$, 
the \emph{categorical rewriting systems for graphs based on sesqui-pushouts},
denoted as $\rs_{\sqpo,i}$, 
is the composition of $\rs_{\fpbc,i}$ and $\rs_{\po,\catgraph}$
(from definition~\ref{definition:crs-po}). 
\end{definition}

\begin{lemma}
\label{lemma:graph-sqpo} 
Let us consider the categorical rewriting system $\rs_{\fpbc,1}$. 
Let $l:K \monomo L$ be a rule and $f_1:L \mo L_1$ a match  
which is conflict-free with respect to $l$. 
Let $(K_1,l_1,g_1)$ be the final pullback complement of $l$ and $f_1$. 
Let $f_2:L_1 \to L_2$ be a match  
which is conflict-free with respect to $l_1$. 
Then $f_2 \circ f_1$ is conflict-free with respect to $l$. 
\end{lemma}

\begin{proof}
Let $f=f_2 \circ f_1:L \mo L_2$. 
The proof is done by contradiction. 
Let us assume that there are two items $x$ and $y$ in $L$
such that $f(x)=f(y)$, with $x\in l(K)$ and $y\not\in l(K)$. 
Then there are two cases:
  \begin{enumerate} 
  \item If $f_1(x)=f_1(y)$ 
    then $f_1$ is not conflict-free with respect to $l$.
  \item Otherwise let $x_1=f_1(x)$ and $y_1=f_1(y)$, so that $f_2(x_1)=f_2(y_1)$. 
    The commutativity of the square $\funS_{\fpbc,l}(f_1)$ 
    implies that $x_1 \in l_1(K_1)$. 
    Moreover, the construction of the final pullback complement
    in \cite[construction 6]{CorradiniHHK06} shows that 
    $y_1 \not \in l_1(K_1)$ since $y \not \in l(K)$. 
    Thus, $x_1\in l_1(K_1)$ and $y_1\not\in l_1(K_1)$,
    so that $f_2$ is not conflict-free with respect to $l_1$.
  \end{enumerate}
\end{proof}

\begin{proposition}
\label{proposition:graph-sqpo} 
The categorical rewriting systems $\rs_{\fpbc,i}$ and $\rs_{\sqpo,i}$,
for $i=1$ and $i=2$, are functorial.
\end{proposition}

\begin{proof} 
Similar to the proof of proposition~\ref{proposition:graph-dpo}. 
\end{proof} 
 
A similar result (vertical composition of
sesqui-pushout graph transformations) is stated in 
\cite[proposition~5]{Lowe10}.

\section{A non-functorial graph transformation system}
\label{sec:garbage}


We define two \emph{garbage removal rewriting systems}, 
as two attempts to formalize
the process of removing unreachable nodes from a given graph.
One of these rewriting systems is not functorial, but the other is. 
Let $\catgraphi$ be the category of graphs with inclusions;
it is a preorder, thus every diagram in $\catgraphi$ is commutative. 
In both rewriting systems, the underlying span is the 
inverse arrows-based span on $\catgraphi$ 
with rules and matches in $\catgraphi$.

\subsection{Garbage removal}
\label{sec:garbage-garbage}

\begin{definition}
Let $L_1$ be a graph and $A$ a subgraph of $L_1$.
The set of nodes of $L_1$ which are \emph{reachable from $A$} ($A$ stands 
for $A$live nodes) 
is defined recursively, as follows: a node of $A$ is reachable from $A$, 
and the successors of a node reachable from $A$ are reachable from $A$. 
The subgraph of $L_1$ generated by the nodes reachable from $A$
is called the \emph{maximal subgraph of $L_1$ reachable from $A$},
it is denoted as $\gr{A}{L_1}$. 
\end{definition}

The aim of garbage removal is the determination of $\gr{A}{L_1}$. 
In fact, $\gr{A}{L_1}$ does not depend on the edges of $A$,
only on its nodes. The nodes of $A$ play the role of \emph{roots} 
for the graph $L_1$, 
with $\gr{A}{L_1}$ as the result of garbage removal from these roots.  
There are several categorical characterizations of $\gr{A}{L_1}$,
see for instance \cite{DuvalEP07}, but they are not used in this paper. 
Garbage removal provides a factorization 
of the inclusion $A\subseteq L_1$ in two inclusions 
$A\subseteq \gr{A}{L_1} \subseteq L_1$. This is denoted:  
  $$ \xymatrix@C=2pc@R=1.5pc{
  A \ar[d]  \\ 
  L_1  \\ 
  }\quad
\xymatrix@R=.8pc{ \\ \ar@{|->}[r]^{GC} & \\}  \quad
\xymatrix@C=2pc@R=1.5pc{
  A \ar[d] \ar[dr] & \ar@{}[dl]|(.7){GC} \\ 
  L_1 & \gr{A}{L_1} \ar[l] \\ 
  }$$
This ``triangular'' diagram is equivalent to the ``rectangular'' one: 
$$ \xymatrix@C=2pc@R=1.5pc{
A \ar[d] & & A \ar[ll] \\ 
L_1 &&  \\ 
} \quad
\xymatrix@R=.8pc{ \\ \ar@{|->}[r]^{GC} & \\}  \quad
\xymatrix@C=2pc@R=1.5pc{
A \ar[d] 
& \ar@{}[d]|{GC} & A \ar[d] \ar[ll] \\ 
L_1 && \gr{A}{L_1} \ar[ll] 
\\ 
}$$

\begin{example} 
\label{example:garbage}
Here are two simple examples, where $A$ is made of a single node. 
  $$
  \begin{array}{r|c|c|c|l}
  \cline{2-2} \cline{4-4}  
  A & a &  
  \xymatrix@R=0pc@C=1pc{ & \ar[l] } & 
  a & A \\
  \cline{2-2} \cline{4-4} 
  \multicolumn{1}{c}{} & 
  \multicolumn{1}{c}{\xymatrix@R=.8pc{ \ar[d] \\ \\ }} & 
  \multicolumn{1}{c}{\xymatrix@R=.8pc{ \ar@{}[d]|{GC} \\ \\ }} & 
  \multicolumn{1}{c}{\xymatrix@R=.8pc{ \ar[d] \\ \\ }}  & 
  \multicolumn{1}{c}{} \\ 
  \cline{2-2} \cline{4-4} 
  L_1 & \xymatrix@R=.7pc@C=1pc{a \ar[d] & b \\ c & \\ } & 
  \xymatrix@R=.5pc@C=1pc{ \\ & \ar[l] \\ } & 
  \xymatrix@R=.7pc@C=1pc{a \ar[d]  \\ c \\ }  &
  \gr{A}{L_1} \\  
  \cline{2-2} \cline{4-4}
  \end{array} 
  \qquad 
  \begin{array}{r|c|c|c|l}
  \cline{2-2} \cline{4-4}  
  A & a &  
  \xymatrix@R=0pc@C=1pc{ & \ar[l] } & 
  a & A \\
  \cline{2-2} \cline{4-4} 
  \multicolumn{1}{c}{} & 
  \multicolumn{1}{c}{\xymatrix@R=.8pc{ \ar[d] \\ \\ }} & 
  \multicolumn{1}{c}{\xymatrix@R=.8pc{ \ar@{}[d]|{GC} \\ \\ }} & 
  \multicolumn{1}{c}{\xymatrix@R=.8pc{ \ar[d] \\ \\ }}  & 
  \multicolumn{1}{c}{} \\ 
  \cline{2-2} \cline{4-4} 
  L_2 & \xymatrix@R=.7pc@C=1pc{a \ar[d] \ar[dr] && b \ar[d] \\ c & d & e \\} & 
  \xymatrix@R=.5pc@C=1pc{ \\ & \ar[l] \\ } & 
  \xymatrix@R=.7pc@C=1pc{a \ar[d] \ar[dr] &  \\ c & d \\ }  &
  \gr{A}{L_2} \\  
  \cline{2-2} \cline{4-4}
  \end{array} 
  $$
\end{example}

We generalize this situation by allowing the rules to be any  
inclusions $R\subseteq L$, not only identities; 
thus for instance the inclusion $\gr{A}{L_1}\subseteq L_1$ can be seen as a rule.
Then, garbage removal can be seen as a categorical rewriting system 
with respect to the 
inverse arrows-based span on $\catgraphi$ 
with rules and matches in $\catgraphi$. 
This can be done in two ways:
in section~\ref{subsec:garbage-lgr} the alive subgraph $A$ 
is the left-hand side~$L$ 
while in section~\ref{subsec:garbage-rgr} it is the right-hand side $R$.

\subsection{Garbage removal as a non-functorial graph rewriting system}
\label{subsec:garbage-lgr}

\begin{definition}
\label{definition:garbage-lgr}
The \emph{$L$-garbage removal rewriting system} $\rs_{\lgr}$ is defined as 
the inverse arrows-based span on $\catgraphi$ 
with rules and matches in $\catgraphi$ 
together with the total functions $\funS_{\lgr,\rho}$, for every $\rho:R\subseteq L$, 
which map each inclusion $L\subseteq L_1$
to the commutative square in $\catgraphi$ 
with vertices $L$, $R$, $L_1$ and $\gr{L}{L_1}$. 
$$ \xymatrix@C=2pc@R=1.5pc{
L \ar[d]_{f} & & R \ar[ll]_{\rho} \\ 
L_1 &&  \\ 
} \quad
\xymatrix@R=.8pc{ \\ \ar@{|->}[r]^{\funS_{\lgr,\rho}} & \\}  \quad
\xymatrix@C=2pc@R=1.5pc{
L \ar[d] \ar[rrd]  
& & R \ar[d] \ar[ll] \ar@{}[lld]|(.3){=}|(.7){GC} \\ 
L_1 && \gr{L}{L_1} \ar[ll] \\ 
}$$
\end{definition}

\begin{proposition}
\label{proposition:garbage-lgr}
The categorical rewriting system $\rs_{\lgr}$ is not functorial.
\end{proposition}

\begin{proof}
In general $\gr{L_1}{L_2}$ is not the same as $\gr{L}{L_2}$,
see example~\ref{example:garbage-lgr} below.
\end{proof}

$$
\xymatrix@C=2pc@R=1.5pc{
L \ar[d] \ar[rrd] 
& & R \ar[d] \ar[ll] \ar@{}[lld]|(.3){=}|(.7){GC} \\ 
L_1 \ar[d] \ar[rrd] 
& & R_1=\gr{L}{L_1} \ar[d] \ar[ll] \ar@{}[lld]|(.3){=}|(.7){GC} \\ 
L_2 && R_2=\gr{L_1}{L_2}  \ar[ll] 
\\  } 
\quad \xymatrix@C=2pc@R=2pc{ \\ \ne \\ } \quad 
\xymatrix@C=2pc@R=2pc{
L \ar[dd] \ar[rrdd] 
& & R \ar[dd] \ar[ll] \ar@{}[lldd]|(.3){=}|(.7){GC} \\ 
 \\ 
L_2 && \gr{L}{L_2} \ar[ll] 
\\  } 
$$

\begin{example} 
\label{example:garbage-lgr}
Let us apply $\rs_{\lgr}$ to $R=L\subseteq L_1 \subseteq L_2$ 
and to $R=L\subseteq L_2$, as in example~\ref{example:garbage}.
We get $\gr{L_1}{L_2} \ne \gr{L}{L_2}$.

  $$
  \begin{array}{r|c|c|c|l}
  \cline{2-2} \cline{4-4}  
  L & a &  
  \xymatrix@R=0pc@C=.6pc{ & \ar[l] } & 
  a & R \\
  \cline{2-2} \cline{4-4} 
  \multicolumn{1}{c}{} & 
  \multicolumn{1}{c}{\xymatrix@R=.8pc{ \ar[d] \\ \\ }} & 
  \multicolumn{1}{c}{} & 
  \multicolumn{1}{c}{\xymatrix@R=.8pc{ \ar[d] \\ \\ }}  & 
  \multicolumn{1}{c}{} \\ 
  \cline{2-2} \cline{4-4} 
  L_1 & \xymatrix@R=.7pc@C=.6pc{a \ar[d] & b \\ c & \\ } & 
  \xymatrix@R=.5pc@C=.6pc{ \\ & \ar[l] \\ } & 
  \xymatrix@R=.7pc@C=.6pc{a \ar[d]  \\ c \\ }  &
  \gr{L}{L_1} \\  
  \cline{2-2} \cline{4-4}
  \multicolumn{1}{c}{} & 
  \multicolumn{1}{c}{\xymatrix@R=.8pc{ \ar[d] \\ \\ }} & 
  \multicolumn{1}{c}{} & 
  \multicolumn{1}{c}{\xymatrix@R=.8pc{ \ar[d] \\ \\ }}  & 
  \multicolumn{1}{c}{} \\ 
  \cline{2-2} \cline{4-4} 
  L_2 & \xymatrix@R=.7pc@C=.6pc{a \ar[d] \ar[dr] && b \ar[d] \\ c & d & e \\} & 
  \xymatrix@R=.5pc@C=.6pc{ \\ & \ar[l] \\ } & 
  \xymatrix@R=.7pc@C=.6pc{a \ar[d] \ar[dr] && b \ar[d] \\ c & d & e \\ }  &
  \gr{L_1}{L_2} \\  
  \cline{2-2} \cline{4-4}
  \end{array} 
  \xymatrix@C=2pc@R=0pc{ \\ \ne \\ } 
  \begin{array}{r|c|c|c|l}
  \cline{2-2} \cline{4-4}  
  L & a &  
  \xymatrix@R=0pc@C=.6pc{ & \ar[l] } & 
  a & R \\
  \cline{2-2} \cline{4-4} 
  \multicolumn{1}{c}{} & 
  \multicolumn{1}{c}{\xymatrix@R=2.8pc{ \ar[d] \\ \\ }} & 
  \multicolumn{1}{c}{} & 
  \multicolumn{1}{c}{\xymatrix@R=2.8pc{ \ar[d] \\ \\ }}  & 
  \multicolumn{1}{c}{} \\ 
  \cline{2-2} \cline{4-4} 
  L_2 & \xymatrix@R=.7pc@C=.6pc{a \ar[d] \ar[dr] && b \ar[d] \\ c & d & e \\} & 
  \xymatrix@R=.5pc@C=.6pc{ \\ & \ar[l] \\ } & 
  \xymatrix@R=.7pc@C=.6pc{a \ar[d] \ar[dr] & \\ c & d  \\ }  &
  \gr{L}{L_2} \\  
  \cline{2-2} \cline{4-4}
  \end{array} 
  $$
\end{example} 

It turns out that if we choose the right-hand side of the rule 
instead of its left-hand side as the alive subgraph, 
the graph transformation system obtained is functorial:
this is done in the next section. 

\subsection{Garbage removal as a functorial graph rewriting system}
\label{subsec:garbage-rgr}

\begin{definition}
\label{definition:garbage-rgr}
The \emph{$R$-garbage removal rewriting system} $\rs_{\rgr}$ is defined as 
the inverse arrows-based span on $\catgraphi$ 
with rules and matches in $\catgraphi$ 
together with the total functions $\funS_{\rgr,\rho}$, for every $\rho:R\subseteq L$, 
which map each inclusion $L\subseteq L_1$
to the commutative square in $\catgraphi$ 
with vertices $L$, $R$, $L_1$ and $\gr{R}{L_1}$. 
$$ \xymatrix@C=2pc@R=1.5pc{
L \ar[d]_{f} & & R \ar[ll]_{\rho} \\ 
L_1 &&  \\ 
} \quad
\xymatrix@R=.8pc{ \\ \ar@{|->}[r]^{\funS_{\rgr,\rho}} & \\}  \quad
\xymatrix@C=2pc@R=1.5pc{
L \ar[d]_{f} \ar@{}[rrd]|(.3){=}|(.7){GC}
& & R \ar[d] \ar[ll]_{\rho} \ar[lld] \\ 
L_1 && \gr{R}{L_1}  \ar[ll] 
\\ 
}$$
\end{definition}

\begin{proposition}
\label{proposition:garbage-rgr}
The categorical rewriting system $\rs_{\rgr}$ is functorial.
\end{proposition}

\begin{proof}
It is easy to check that $\gr{R_1}{L_2}$, where $R_1=\gr{R}{L_1}$, 
is the same as $\gr{R}{L_2}$.
\end{proof}

$$
\xymatrix@C=2pc@R=1.5pc{
L \ar[d] \ar@{}[rrd]|(.3){=}|(.7){GC}
& & R \ar[d] \ar[ll] \ar[lld] \\ 
L_1 \ar[d] \ar@{}[rrd]|(.3){=}|(.7){GC}
& & R_1=\gr{R}{L_1} \ar[d] \ar[ll] \ar[lld] \\ 
L_2 && R_2=\gr{R_1}{L_2}  \ar[ll] 
\\  } 
\quad \xymatrix@C=2pc@R=2pc{ \\ = \\ } \quad 
\xymatrix@C=2pc@R=2pc{
L \ar[dd] \ar@{}[rrdd]|(.3){=}|(.7){GC}
& & R \ar[dd] \ar[ll] \ar[lldd] \\ 
 \\ 
L_2 && R_2=\gr{R}{L_2}  \ar[ll] 
\\  } 
$$

\begin{example} 
\label{example:garbage-rgr}
Let us apply $\rs_{\rgr}$ to $R=L\subseteq L_1 \subseteq L_2$ 
and to $R=L\subseteq L_2$, as in example~\ref{example:garbage}.
We get $\gr{R_1}{L_2} = \gr{R}{L_2}$.
  $$
  \begin{array}{r|c|c|c|l}
  \cline{2-2} \cline{4-4}  
  L & a &  
  \xymatrix@R=0pc@C=1pc{ & \ar[l] } & 
  a & R \\
  \cline{2-2} \cline{4-4} 
  \multicolumn{1}{c}{} & 
  \multicolumn{1}{c}{\xymatrix@R=.8pc{ \ar[d] \\ \\ }} & 
  \multicolumn{1}{c}{} & 
  \multicolumn{1}{c}{\xymatrix@R=.8pc{ \ar[d] \\ \\ }}  & 
  \multicolumn{1}{c}{} \\ 
  \cline{2-2} \cline{4-4} 
  L_1 & \xymatrix@R=.7pc@C=1pc{a \ar[d] & b \\ c & \\ } & 
  \xymatrix@R=.5pc@C=1pc{ \\ & \ar[l] \\ } & 
  \xymatrix@R=.7pc@C=1pc{a \ar[d]  \\ c \\ }  &
  \xymatrix@R=0pc@C=0pc{R_1\!\!=\;\;\;\; \\ \gr{R}{L_1} \\ } \\  
  \cline{2-2} \cline{4-4}
  \multicolumn{1}{c}{} & 
  \multicolumn{1}{c}{\xymatrix@R=.8pc{ \ar[d] \\ \\ }} & 
  \multicolumn{1}{c}{} & 
  \multicolumn{1}{c}{\xymatrix@R=.8pc{ \ar[d] \\ \\ }}  & 
  \multicolumn{1}{c}{} \\ 
  \cline{2-2} \cline{4-4} 
  L_2 & \xymatrix@R=.7pc@C=1pc{a \ar[d] \ar[dr] && b \ar[d] \\ c & d & e \\} & 
  \xymatrix@R=.5pc@C=1pc{ \\ & \ar[l] \\ } & 
  \xymatrix@R=.7pc@C=1pc{a \ar[d] \ar[dr] &  \\ c & d \\ }  &
  \gr{R_1}{L_2} \\  
  \cline{2-2} \cline{4-4}
  \end{array} 
  \xymatrix@C=2pc@R=0pc{ \\ = \\ } 
  \begin{array}{r|c|c|c|l}
  \cline{2-2} \cline{4-4}  
  L & a &  
  \xymatrix@R=0pc@C=.6pc{ & \ar[l] } & 
  a & R \\
  \cline{2-2} \cline{4-4} 
  \multicolumn{1}{c}{} & 
  \multicolumn{1}{c}{\xymatrix@R=2.8pc{ \ar[d] \\ \\ }} & 
  \multicolumn{1}{c}{} & 
  \multicolumn{1}{c}{\xymatrix@R=2.8pc{ \ar[d] \\ \\ }}  & 
  \multicolumn{1}{c}{} \\ 
  \cline{2-2} \cline{4-4} 
  L_2 & \xymatrix@R=.7pc@C=.6pc{a \ar[d] \ar[dr] && b \ar[d] \\ c & d & e \\} & 
  \xymatrix@R=.5pc@C=.6pc{ \\ & \ar[l] \\ } & 
  \xymatrix@R=.7pc@C=.6pc{a \ar[d] \ar[dr] & \\ c & d  \\ }  &
  \gr{R}{L_2} \\  
  \cline{2-2} \cline{4-4}
  \end{array} 
  $$
\end{example}

\section{Conclusion}

We have introduced a new notion of abstract rewriting system based on categories. 
These systems are designed for dealing with 
abstract rewriting frameworks where rewrite steps 
are defined by means of matches. We have defined the properties of (horizontal) 
composition as well as functoriality of rewriting in our abstract setting 
and we have illustrated these properties throughout several 
algebraic graph rewriting systems. 
We plan to extend and deepen our abstract framework 
by investigating other instances such as \cite{Heckel97,Lowe10} 
and by allowing the rewriting processes $\funS_{\rho}$
to be relations instead of partial functions. 

\subsection*{Acknowledgements}
We would like to thank Andrea Corradini
and Barbara K{\"o}nig for enlighting discussions about the Dyckhoff-Tholen 
condition. We also thank anonymous referees for insightful comments.



\begin{thebibliography}{15}

\bibitem{BaaderN98}
F.~Baader and T.~Nipkow.
\newblock {\em Term rewriting and all that}.
\newblock Cambridge University Press, 1998.

\bibitem{CorradiniHHK06}
A.~Corradini, T.~Heindel, F.~Hermann, and B.~K{\"o}nig.
\newblock Sesqui-pushout rewriting.
\newblock In {\em Third International Conference on Graph Transformations (ICGT
  06)}, volume 4178 of {\em Lecture Notes in Computer Science}, pages 30--45.
  Springer, 2006.

\bibitem{CorradiniMREHL97}
A.~Corradini, U.~Montanari, F.~Rossi, H.~Ehrig, R.~Heckel, and M.~L{\"o}we.
\newblock Algebraic approaches to graph transformation - part {I}: Basic
  concepts and double pushout approach.
\newblock In {\em Handbook of Graph Grammars}, pages 163--246, 1997.

\bibitem{DuvalEP07}
D.~Duval, R.~Echahed, and F.~Prost.
\newblock Adjunction for Garbage Collection with Application to Graph 
          Rewriting.
\newblock In {\em 18th International Conference on Rewriting Techniques and
  Applications, RTA 2007}, 
  Springer Lecture Notes in Computer Science 4533
  pages 122--136, 2007.

\bibitem{DuvalEP09}
D.~Duval, R.~Echahed, and F.~Prost.
\newblock A heterogeneous pushout approach to term-graph transformation.
\newblock In {\em 20th International Conference on Rewriting Techniques and
  Applications, RTA 2009}, 
  Springer Lecture Notes in Computer Science 5595
  pages 194--208, 2009.

\bibitem{DuvalEP11}
D.~Duval, R.~Echahed, and F.~Prost.
\newblock Graph rewriting with polarized cloning.
\newblock {\em Available at http://arxiv.org/abs/0811.3400 Submitted.}

\bibitem{DyckThol87}
R.~Dyckhoff and W.~Tholen.
\newblock Exponentiable morphisms, partial products and pullback complements.
\newblock In {\em Journal of Pure and Applied Algebra}, 
49(1{\&}2):103--116, 1987.

\bibitem{EcJ98a07}
R.~Echahed and J.~C. Janodet.
\newblock Admissible graph rewriting and narrowing.
\newblock In {\em Proc.~of Joint International Conference and Symposium on
  Logic Programming (JICSLP'98)}, pages 325--340. MIT Press, June 1998.

\bibitem{Ehrig2004}
H.~Ehrig, A.~Habel, J.~Padberg, and U.~Prange.
\newblock Adhesive High-Level  Replacement Categories and Systems.
\newblock In {\em Proc.~of ICGT 2004}, volume 3256 of {\em Lecture
  Notes in Computer Science}, pages 144--160, Springer, 2004. 

\bibitem{EhrigHKLRWC97}
H.~Ehrig, R.~Heckel, M.~Korff, M.~L{\"o}we, L.~Ribeiro, A.~Wagner, and
  A.~Corradini.
\newblock Algebraic approaches to graph transformation - part {II}: Single
  pushout approach and comparison with double pushout approach.
\newblock In {\em Handbook of Graph Grammars}, pages 247--312, 1997.

\bibitem{EhrigPS73}
H.~Ehrig, M.~Pfender, and H.~J. Schneider.
\newblock Graph-grammars: An algebraic approach.
\newblock In {\em 14th Annual Symposium on Foundations of Computer Science
  (FOCS), 15-17 October 1973, The University of Iowa, USA}, pages 167--180.
  IEEE, 1973.

\bibitem{Heckel97}
R.~Heckel, H.~Ehrig, U.~Wolter and A.~Corradini. 
\newblock Double-pullback transitions and coalgebraic loose semantics 
for graph transformation systems.
\newblock In {\em Applied Categorical Structures}, 9:83--110, 
1997. 

\bibitem{Lowe93}
M.~L{\"o}we.
\newblock Algebraic approach to single-pushout graph transformation.
\newblock {\em Theor. Comput. Sci.}, 109(1{\&}2):181--224, 1993.

\bibitem{Lowe10}
M.~L{\"o}we.
\newblock Graph-rewriting in span-categories. 
\newblock In {\em Fith International Conference on Graph Transformations (ICGT
  10)}, volume 6372 of {\em Lecture Notes in Computer Science}, pages 218--233.
  Springer, 2010. 

\bibitem{MacLane}
S.~Mac Lane.
\newblock {\em Categories for the Working Mathematician.} 2nd edition.
\newblock Graduate Texts in Mathematics 5, Springer-Verlag (1997). 

\end{thebibliography}
\end{document}